\documentclass[10pt,twocolumn]{IEEEtran}
\makeatletter
\def\ps@headings{%
\def\@oddhead{\mbox{}\scriptsize\rightmark \hfil \thepage}%
\def\@evenhead{\scriptsize\thepage \hfil \leftmark\mbox{}}%
\def\@oddfoot{}%
\def\@evenfoot{}}
\makeatother
\usepackage{amssymb}
\usepackage{mathrsfs}
\usepackage[dvips]{graphicx}

\def\Route{R}
\def\route{r}
\def\line{l} 
\def\dia{d_{max}}

\def\bfone{{\bf 1}}
\def\intersect{\neq\emptyset}
\def\defeq{\stackrel{\mathrm{def}}{=}}
\def\dD{dx_d\,dy_d\,d\theta}
\def\disconnect{\Pr(s\not \leftrightarrow t)}
\newcommand{\sizex}[1]{\| #1 \|}
\newcommand{\lengthx}[1]{| #1 |}

\newtheorem{lemma}{\bf{Lemma}}
\newtheorem{theorem}{\bf{Theorem}}
\newtheorem{cor}{\bf{Corollary}}
\newtheorem{example}{\bf{Example}}

\title{Spatial Design of Physical Network Robust against Earthquakes}
\author{Hiroshi~Saito,~\IEEEmembership{Fellow,~IEEE} 
\thanks{Manuscript received }
\thanks{Hiroshi Saito is with NTT Network Technology Laboratories, 3-9-11, Midori-cho, Musashino-shi, Tokyo 180-8585, Japan, E-mail: saito.hiroshi@lab.ntt.co.jp, URL: http://www9.plala.or.jp/hslab/, Phone: +81 422 59 4300, Fax: +81 422 59 5671.}}

\date{}

\begin{document}

\maketitle

\begin{abstract}
This paper analyzes the survivability of a physical network against earthquakes and proposes spatial network design rules to make a network robust against earthquakes.
The disaster area model used is fairly generic and bounded.
The proposed design rules for physical networks include: (i) a shorter zigzag route can reduce the probability that a network intersects a disaster area, (ii) an additive performance metric, such as repair cost, is independent of the network shape if the route length is fixed, and (iii) additional routes within a ring network does not decrease the probability that all the routes between a given pair of nodes intersect the disaster area, but a wider detour route decreases it.
Formulas for evaluating the probability of disconnecting two given nodes are also derived.
An optimal server placement is shown as an application of the theoretical results.
These analysis results are validated through empirical earthquake data.
\end{abstract}

\section{Introduction}
Japan is a country prone to earthquakes.
An earthquake occurred on March 11, 2011 in northeastern Japan, and the resulting tsunami caused fatal damage to the network.
In NTT's networks \cite{ntt}, approximately 1.5 million circuits for fixed-line services, approximately 6,700 pieces of mobile base-station equipment, and approximately 15,000 circuits for corporate data communication services were damaged.
It was also reported that transmission lines were disconnected in 90 routes, 18 exchange office buildings were destroyed, 23 buildings were submerged, approximately 65,000 telephone poles were destroyed, and submersion and physical damage to aerial cables reached about 6,300 kilometers. 
Many other countries are also prone to earthquakes \cite{china}.

Through such experiences, network operators have made efforts to increase network robustness against earthquakes.
For example, microwave transit systems were expanded as a means of increasing network survivability after a large earthquake struck in 1968, and a transportable earth station of a satellite communication system was developed because of an earthquake in 1993 \cite{ntt-east}.
Nevertheless, we do not have a design method of creating a network robust against earthquakes.
This paper responds to this need.

The contributions of this paper are as follows. (1) This paper introduces a theoretical method for evaluating metrics such as probability of disconnection for a bounded and general-shaped disaster area.
(2) This theoretical method explicitly reveals physical network design rules robust against earthquakes.
For example, (i) a shorter zigzag route can reduce the probability that a network intersects the disaster area, (ii) an additive performance metric is independent of the network shape if the route length is fixed, and (iii) additional routes within a ring network does not decrease the probability that all the routes between a pair of nodes intersect the disaster area.
(3) Actual earthquake intensity maps are used for disasters areas to evaluate the validity of the theoretical results.

The organization of this paper is as follows. 
Related work is described in Section II,
the model and notations used in this paper is explained in Section III,
the analysis is discussed in Section IV,
numerical examples are discussed in Section V, and
the paper is concluded in Section VI.

\section{Related work}
Although a large number of theoretical papers have been published evaluating the reliability, availability, and survivability for a given network, most of these papers focus on a single failure (or independent failures) of a network node or a link.
For a given set of network topologies and failure rates of network entities, they typically evaluate a metric, such as the probability that a pair of network entities can be connected \cite{book}.
In a disaster, the assumption of a single or independent failure is not valid.
In addition, the physical shape of a network is important for evaluating the impact of a disaster on network survivability; however, most studies have not covered this point.

The following studies focused on network survivability by taking into account correlated failure and geometric/geographical conditions:
Grubesic \cite{geography} evaluated the network survivability of the current Internet based on geographical data.
Although he focused on the physical route of a network, it was a case study, and no mathematical models or methods were provided.
Liew and Lu \cite{survivability} proposed a framework to evaluate network survivability during a disaster and introduced a survivability function to various metrics.
Although their framework can introduce correlated failures, they did not propose any method or model of correlations.
Wu et al. \cite{underseaCableFailure} discussed the optimization of the physical route of an undersea cable by assuming a disk-shaped disaster area.
By assuming a rectangular route, the length of an edge is determined by minimizing cost while maintaining a higher probability of connecting two cities than the threshold.

Another direction is the extension of the minimum-cut-max-flow type problems by taking into account a disaster area.
As far as I know, Bienstock \cite{OR} initiated the study of this problem.
Algorithms computing the minimum number of disaster areas disconnecting the source and sink nodes were investigated when all the edges intersecting the disaster areas are removed.
Sen et al. \cite{sen} proposed a region-based connectivity as a metric for fault-tolerance.
Assuming the region is a disk-shaped disaster area, polynomial time algorithms calculating region-based connectivity are provided.
Neumayer et al. \cite{discMinCut} discussed the geographical min-cut, defined as the minimum number of disk-shaped disaster areas to disconnect a pair of nodes, and the geographical max-flow, defined as the maximum number of paths that are not disconnected by a single disaster area, and showed that geographical min-cut is not equal to geographical max-flow.
Agarwal et al. studied algorithms that find a disaster location having the highest expected impact on a network, where the impact is defined by various metrics such as the number of failed components \cite{wdmFailure}.

Recently, Neumayer et al. published two papers intended to cover network survivability in a disaster \cite{failureToN}, \cite{failureINFOCOM}.
In their network model, there is a set of line segments of which end points are locations of network center buildings and the disaster model is a line segment or a circle \cite{failureToN}.
They proposed to use an optimization technique to find the worst case disaster.
On the other hand, Neumayer and Modiano \cite{failureINFOCOM} used geometric probability (integral geometry) to model the randomness of a disaster.
Their network model is, again, a set of line segments of which end points are locations of network center buildings and the disaster model is a line.
These papers emphasize the polynomial time algorithm to evaluate metrics.
The paper \cite{infocom2014} is in this direction.
Its network model is with nodes and links consisting of line-segments, and the disaster area is assumed to be a half plane.
Nodes and links in a disaster area is probabilistically in failure.
In addition to the algorithm evaluating metrics, this paper proposes an update method and an optimal placement of network entities to make the network robust.

Saito \cite{ToNsaito} derives explicit formulas of various metrics under the assumption that the disaster area is a half plane (or a broad strip) and that the physical network shape is fairly generic.
Based on these formulas, a rule of thumb of network design robust against disasters was proposed in that paper.
For example, (1) reducing the convex hull of the physical route reduces the expected number of nodes that cannot connect to the destination. 
(2) The probability of maintaining the connectivity of two nodes on a ring-type network cannot be changed by changing the physical route of that network. 
(3) The effect of making a ring-type network is identical to that of a single physical route implemented by the straight-line route.

This paper extends \cite{ToNsaito} regarding the following points:
The finite convex disaster (earthquake) area model is assumed. That is, the metrics derived in this paper shows the impact of the disaster area size and the shape of the disaster area. In the numerical example, actual disaster area data are used.
Optimal network design regarding some metrics is explicitly shown.

Although this paper does not require users to have knowledge of integral geometry (geometric probability), derivation of some basic results follows a method used in integral geometry.
In the network and network application communities, integral geometry is not commonly used.
In addition to \cite{failureINFOCOM}, \cite{infocom2014}, and \cite{ToNsaito}, the following papers use it.
For example, a series of papers \cite{infocom}, \cite{mobileComp}, \cite{signal} proposed shape estimation methods derived using integral geometry for a target object based on reports from sensor nodes of unknown locations.
Lazos et al. \cite{detection} and Lazos and Poovendran \cite{lazos} directly applied the results \cite{Santalo} to the analysis of detecting an object moving in a straight line and to the evaluation of the probability of $k$-coverage. Kwon and Shroff \cite{routing} also applied integral geometry to the analysis of straight-line routing, which is an approximation of the shortest path routing, and Choi and Das \cite{energy} used it to determine sensors in energy-conserving data gathering.

\section{Model and notations}
\subsection{Notations}
The following notations are used for the remainder of this paper.
\begin{itemize}
\item $d(u,v)$: the distance between $u$ and $v$.
\item $\line(u,v)$: a line segment between $u$ and $v$.
\item $L(c)$: the length of curve $c$.
\item $C(c_1,c_2)$: a set of interior points bounded by two curves $c_1$ and $c_2$, where the two end points of $c_1$ are identical to those of $c_2$. Its boundary is $c_1\cup c_2$.
\item $\lengthx{A}$: the perimeter length of a bounded area $A$.
\item $\sizex{A}$: the size of a bounded set $A$.
(When $A$ is defined in the parameter space $(x_d,y_d,\theta)$, $\sizex{A}\defeq
\int_A\dD$.
When $A$ is defined in $\mathbb{R}^2$, $\sizex{A}\defeq\int_A dx\,dy$.)
\item $\dia (A)$: the diameter of a bounded area $A$.
That is, $\dia (A)=\max_{u,v\subset A}d(u,v)$.
\end{itemize}

\subsection{Model}
Let $\route(s,t)\subset A_0\subset \mathbb{R}^2$ be a physical route between a pair of nodes $s$ and $t$, where $A_0$ is an area of interest and is bounded and convex.
(For simplicity, $(s,t)$ may be removed in $\route(s,t)$.)
In this paper, the meaning of a route $\route$ is not limited to the connectivity between $s$ and $t$.
The meaning of $\route$ implies the physical route shape.
When there are two routes between $s$ and $t$ and they are disjoint except for $s$ and $t$, we say that $s$ and $t$ are on a ring-type network and that these two routes form a ring-type network.
In the analysis in this paper, $\disconnect$ (the probability disconnecting $s$ and $t$) and $\Pr(\route\cap D\neq \emptyset)$ are mainly discussed, where $D$ is explained below.

Let $D$ be a disaster area caused by an earthquake.
Assume that $D$ is convex in the remainder of this paper if not explicitly indicated otherwise.
The position of $D$ is characterized by the position of its reference point $(x_d,y_d)$ and by its direction $\theta$ formed by a reference line fixed to $D$ with another reference line fixed to the fixed coordinates (Fig. \ref{disasterArea}).
The disaster area is modeled as a randomly placed area affecting $A_0$.
That is, $(x_d,y_d,\theta)$ is in $\Omega(A_0)$ in the parameter space where $\Omega(A_0)\defeq\{(x_d,y_d,\theta)|D(x_d,y_d,\theta)\cap A_0\intersect\}$.

\begin{figure}[htb] 
\begin{center} 
\includegraphics[width=8cm,clip]{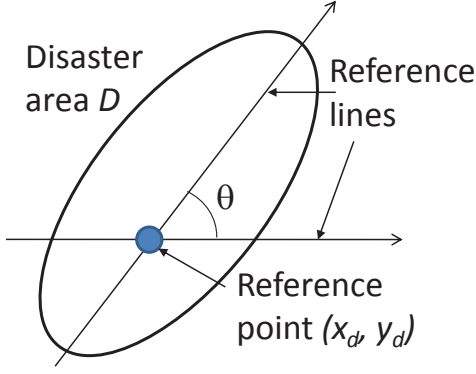} 
\caption{Disaster area model} 
\label{disasterArea} 
\end{center} 
\end{figure}

Assume that $(x_d,y_d)$ and $\theta$ are uniformly distributed in $\Omega(A_0)$ because we have no prior information regarding the location of $D$.
We can define a probability that a set of positions of $D$ satisfies a certain condition $X_c$.
Because $(x_d,y_d)$ and $\theta$ are uniformly distributed in $\Omega(A_0)$, this probability is given by the ratio of the size of the subspace $X(X_c)\defeq\{(x_d,y_d,\theta)|(x_d,y_d,\theta)$ satisfies $X_c\}$ to the size of $\Omega(A_0)$.
That is, the probability that a set of positions of $D$ satisfies $X_c$ is $\sizex{X(X_c)}/\sizex{\Omega(A_0)}$.
(This is formally called a geometric probability based on integral geometry \cite{Santalo}.)
For convex $D$ and convex $A_0$, it is known that $\sizex{\Omega(A_0)}$ is given as follows (Eq. (6.48) in \cite{Santalo}).
\begin{equation}\label{size_two_convex}
\sizex{\Omega(A_0)}=2\pi(\sizex{D}+\sizex{A_0})+\lengthx{D}\cdot\lengthx{A_0}
\end{equation}

\begin{example}
Assume that $D$ is a disk of radius $r_d$, $A_0$ is a disk of radius $r_0$, and condition $X_c$ is $\{(x_d,y_d,\theta)|D(x_d,y_d,\theta)\cap \line(s,t)\intersect \}$.
For this example, we can easily draw a picture because $D$ is independent of $\theta$ (Fig. \ref{simpleExample}).
Note that the position of reference point $(x_d,y_d)$ must be in $X(X_c)$ to satisfy $D\cap \line(s,t)\intersect$.
Because $(x_d,y_d)$ is uniformly distributed in $\Omega(A_0)$ if no condition is specified, the probability that a set of positions of $D$ satisfies $X_c$ is $\int_{X(X_c)}\dD/\int_{\Omega(A_0)}\dD$.
Here, $\int_{X(X_c)}\dD=2\pi(\pi r_d^2+2d(s,t)r_d)$ and $\int_{\Omega(A_0)}\dD=2\pi^2(r_0+r_d)^2$.

\begin{figure}[htb] 
\begin{center} 
\includegraphics[width=8cm,clip]{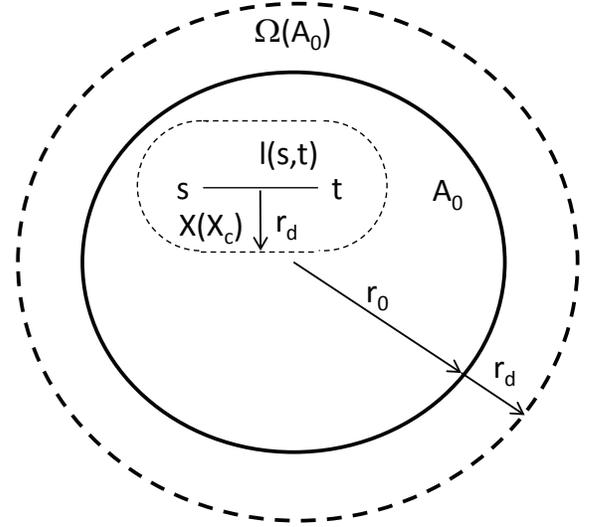} 
\caption{Simple example of $\{(x_d,y_d,\theta)|D(x_d,y_d,\theta)\cap \line(s,t)\intersect \}$} 
\label{simpleExample} 
\end{center} 
\end{figure}

\end{example}

Similarly to the definition of probability, we can also define an expectation.
For each position $(x_d,y_d,\theta)$ of $D$, we can define a quantity $q(x_d,y_d,\theta)$.
An example of $q(x_d,y_d,\theta)$ is the length of chord $D\cap \line(s,t)$.
Because $(x_d,y_d)$ and $\theta$ are uniformly distributed in $\Omega(A_0)$, the probability of the position $[x_d,x_d+dx_d)\times[y_d,y_d+dy_d)\times[\theta,\theta+d\theta)$
is $\dD/\sizex{\Omega(A_0)}$.
Therefore, the expectation of the quantity $q$ can be defined by $\int_{\Omega(A_0)}q(x_d,y_d,\theta)\dD/\sizex{\Omega(A_0)}$.

\section{Analysis}
\subsection{Probability of routes intersecting disaster area}
In this subsection, the probability that routes intersect $D$ is analyzed.

\subsubsection{Single route}
\begin{lemma}\label{lemma1}
Assume that route $\route$ consists of two line segments $l_1,l_2$ connecting to an inner angle $\phi\leq \pi$, and that the distance between an end point not included in $l_i$ and $l_i$ is larger than $\dia (D)$ for $i=1,2$.
When $L(\route)$ is fixed, $\int_{\route\cap D\intersect}\dD$ is given by $2\lengthx{D}L(\route)+2\pi \sizex{D}-f(\phi)$, where $f(\phi)$ is a decreasing function of $\phi\leq \pi$, $f(\phi)\geq 0$.
\end{lemma}
\begin{proof}
For fixed $\theta$, $(x_d,y_d)$ satisfying $\{\route\cap D\intersect\}$ is shown in Figure \ref{fixed_theta}.
As shown in this figure, $\int_{\route\cap D\intersect}dx_d\,dy_d$ consists of seven parts.
The first four parts are four parallelograms located on both sides of the two line segments
(shown with black dotted lines in Fig. \ref{fixed_theta}.)
It should be noted that they overlap.
The overlap area is hatched in this figure.

The remaining three parts are associated with three vertexes including two end points.
At each of the three convex vertexes (i)-(iii), $D$ comes in contact with the vertex and the reference point of $D$ draws a curve.
This curve and the two line segments from the vertex to the reference point form a sector-like shaped area.
If we fix $D$ and move vertexes, these curves become parts of the boundary of $D$, and these sector-like shaped areas become parts of $D$.
In the window in Fig. \ref{fixed_theta}, the parts of $D$ corresponding to these sector-like shaped areas are formed at vertexes (i)-(iii).
Because there is a concave vertex (ii'), the sum of these parts are not identical to $D$.
In Fig. \ref{fixed_theta}, the surplus area, which is the sum of the parts of $D$ corresponding to these sector-like shaped areas formed at vertexes (i)-(iii) minus $D$, is the sector-like shaped area specified by two tangent points, each of which is formed by $D$ and each of the two line segments $l_1,l_2$.

\begin{figure}[htb] 
\begin{center} 
\includegraphics[width=8cm,clip]{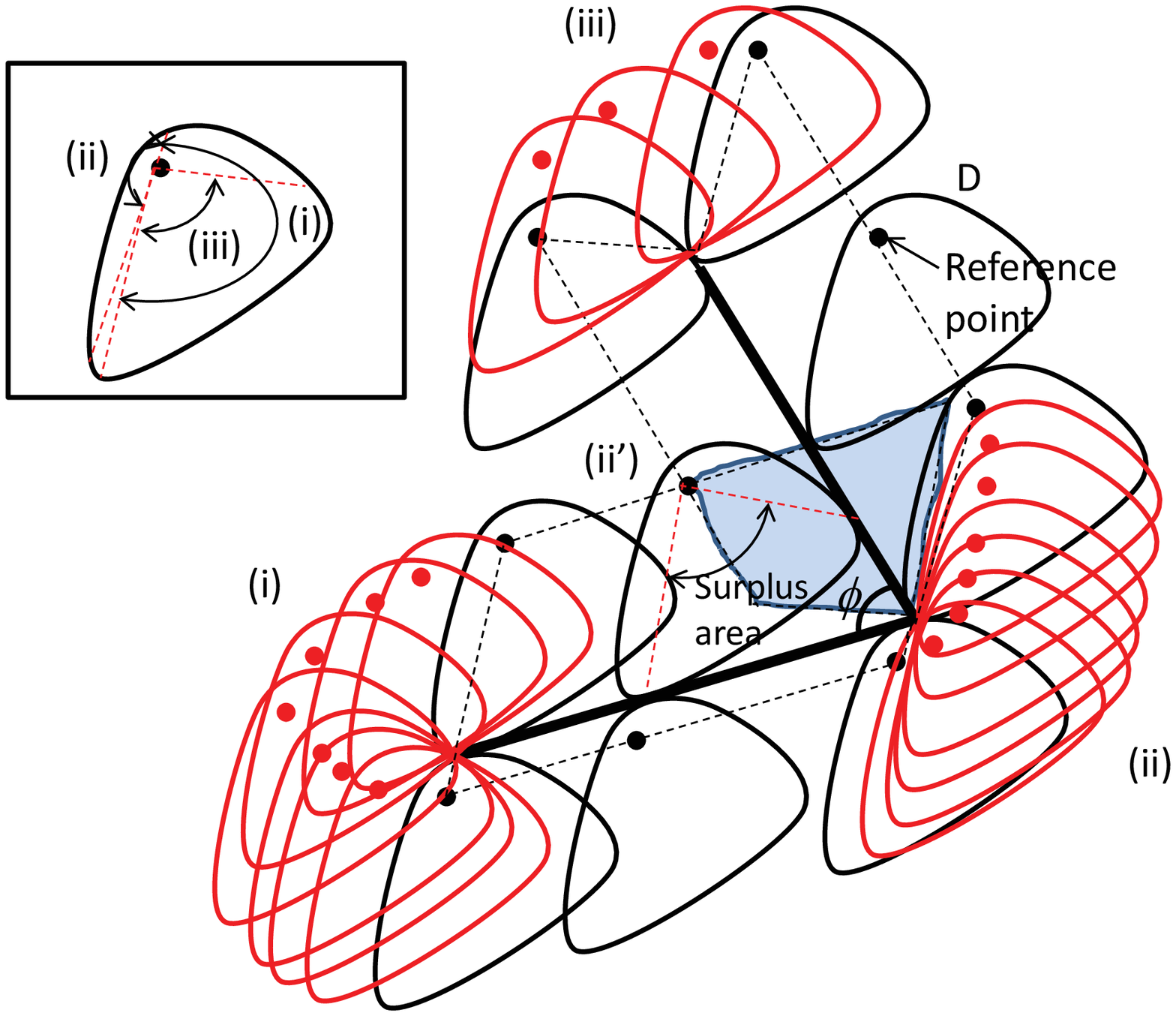} 
\caption{Illustration of $\int_{\route\cap D\intersect}dx_d\,dy_d$} 
\label{fixed_theta} 
\end{center} 
\end{figure}

In summary, $\int_{\route_1\cap D\intersect}dx_d\,dy_d$ is equal to the sum of the size of the four parallelograms, $D$, and the surplus area, minus the overlap area size.
Each side length of these parallelograms is the length of $l_1$ or $l_2$, and its height is the distance from reference point $(x_d,y_d)$ to each of these two line segments.
In fact, these two line segments are on the tangent lines of $D$.
Each of these tangent lines is called a line of support of $D$ \cite{dictionary}.
For the distance $\rho$ from the reference point to the line of support of $D$, the following equation is known: $\int_{-\pi}^\pi \rho(\theta) d\theta=\lengthx{D}$ (Eq. (1.6) in \cite{Santalo}).

Figure \ref{fixed_theta_overlap}-(a) focuses on the surplus and overlap areas.
The former is included in the area outlined with thick (red) dotted lines, and the latter is outlined with thick (black) lines.
Because the area outlined with thick (red) dotted lines and the overlap area are congruent, the size of the overlap area minus the size of the surplus area is equal to that of the hatched area in Fig. \ref{fixed_theta_overlap}-(a).
Let $f(\xi,\phi|\theta)$ be this size of the hatched area in Fig. \ref{fixed_theta_overlap}-(a), where $\xi$ is the angle specifying a tangent point and is a function of $\phi$ and $\theta$.
As shown in Fig. \ref{fixed_theta_overlap}-(b), this size decreases as $\phi$ increases.
Therefore, $f(\xi,\phi|\theta)$ is larger than or equal to 0, and a decreasing function of $\phi$ for each $\theta$.

\begin{figure}[htb] 
\begin{center} 
\includegraphics[width=8cm,clip]{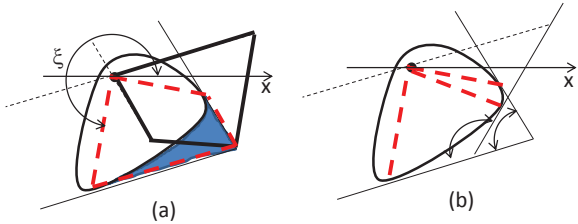} 
\caption{Illustration of surplus and overlap areas} 
\label{fixed_theta_overlap} 
\end{center} 
\end{figure}

Note that $\int_{\route\cap D\intersect}\dD$ is the integral on $\theta$ of the size of the four parallelograms and $D$ minus $f(\xi,\phi|\theta)$.
First, the integral of the size of the four parallelograms on $\theta$ is equal to $\sum_{i=1}^2 \int_{-\pi}^\pi(\rho_i(\theta)+\rho_i(\theta+\pi))L(l_i)d\theta=2\lengthx{D}\sum_iL(l_i)=2\lengthx{D}L(\route)$.
Second, the integral of the size $D$ on $\theta$ is $2\pi \sizex{D}$.
Third, $f(\phi)\defeq \int_{-\pi}^\pi f(\xi,\phi|\theta)d\theta$ is larger than or equal to 0, and a decreasing function of $\phi$.

The assumption that the distance between an end point not included in $l_i$ and $l_i$ is larger than $\dia (D)$ is required because $D$ comes into contact with $l_1$ and $l_2$ at the tangent points at (ii').
\end{proof}

Here is a simple example of this lemma.
\begin{example}
Assume that $D$ is a disk of radius $r_d$.
Figure \ref{1vertexEx}-(a) shows $\{(x_d,y_d)|D\cap\route \intersect\}$ and Fig. \ref{1vertexEx}-(b) highlights the overlap area, which is the quadrangle in red.
As shown in this figure, $\int_{D\cap\route \intersect}dx_d\,dy_d=2(L(l_1)+L(l_2))r_d+(3\pi/2-\phi/2)r_d^2-r_d^2/\tan(\phi/2)$.
Therefore, $\int_{\route\cap D\intersect}\dD=4\pi(L(l_1)+L(l_2))r_d+\pi(3\pi-\phi)r_d^2-2\pi r_d^2/\tan(\phi/2)$.
Hence, $f(\phi)=-\pi(\pi-\phi)r_d^2+2\pi r_d^2/\tan(\phi/2)$.

\begin{figure}[htb] 
\begin{center} 
\includegraphics[width=8cm,clip]{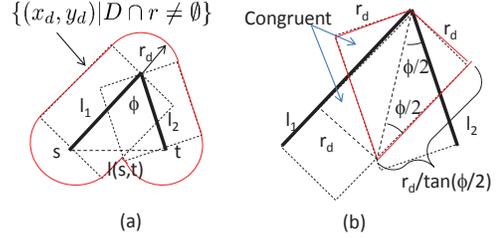} 
\caption{Simple example of Lemma \ref{lemma1}} 
\label{1vertexEx} 
\end{center} 
\end{figure}
\end{example}




Lemma \ref{lemma1} can be generalized.
\begin{lemma}\label{lemma2}
Assume that $\route$ consists of $n$ line segments, and that the distance between the $i$-th line segment and any point on the $j$-th line segments is larger than $\dia (D)$ for $j\neq i-1,i,i+1$ and for all $i$.
(One of two end points of $\route$ is considered as a part of the 0-th and first line segments, and the other one is a part of the $n$-th and $(n+1)$-th line segments for simplifying the notation.)
Let $\phi_i\leq \pi$ be the inner angle of the $i$-th and $(i+1)$-th line segments ($i=1,2\cdots,n-1$).
When $L(\route)$ is fixed, $\int_{\route\cap D\intersect}\dD$ is given by $2\lengthx{D}L(\route)+2\pi \sizex{D}-\sum_i f(\phi_i)$, where $f(\phi_i)$ is a decreasing function of $\phi_i\leq \pi$, $f(\phi_i)\geq 0$.
\end{lemma}
\begin{proof}
For fixed $\theta$, $(x_d,y_d)$ satisfying $\{\route\cap D\intersect\}$ is shown in Fig. \ref{fixed_theta_multi}.
As shown in this figure, $\int_{\route\cap D\intersect}dx_d\,dy_d$ consists of several parts in two categories.
The first category is associated with parallelograms located on both sides of the line segments.
For the concave vertexes, their overlap areas appear.
They are hatched and denoted as overlap-(x') in the figure.

\begin{figure}[htb] 
\begin{center} 
\includegraphics[width=8cm,clip]{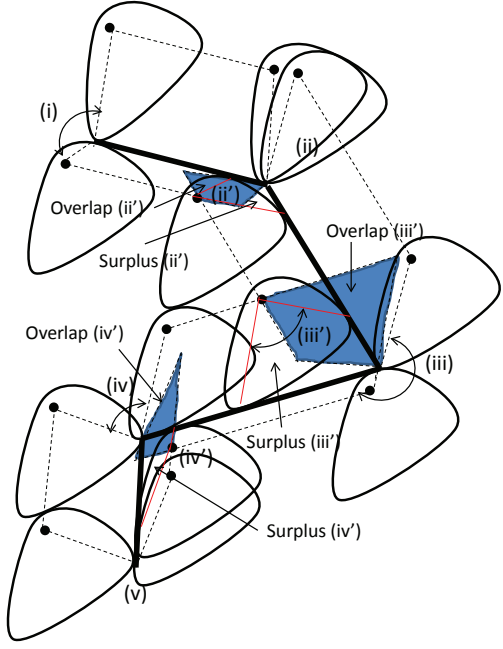} 
\caption{Illustration of $\int_{\route\cap D\intersect}dx_d\,dy_d$: $\route$ consisting of more than two line segments} 
\label{fixed_theta_multi} 
\end{center} 
\end{figure}

The second category is associated with convex vertexes including end points of $\route$.
At each vertex, $D$ comes into contact with the vertex and the reference point of $D$ forms a curve.
This curve and the two line segments from the vertex to the reference point form a sector-like shaped area.
If we fix $D$ and move the vertex, this curve becomes a part of the boundary of $D$ and this sector-like shaped area becomes a part of $D$.
In Fig. \ref{fixed_theta_angle}, parts of $D$ corresponding to these sector-like shaped areas are described.
When $D$ starts around (i) in Fig. \ref{fixed_theta_multi}, it moves through (ii), (iii), (iv'), (v), (iv), (iii'), (ii'), and finally returns to (i). The trace on the boundary of $D$ corresponding to this movement starts from ^^ ^^ Start" in Fig. \ref{fixed_theta_angle}, moves clockwise in the sector-like shaped areas (i), (ii), and (iii), counterclockwise in (iv'), clockwise in (v) and (iv), and counter-clockwise in (iii'), (ii'), and finally returns to ^^ ^^ Start."
Therefore, the sum of these sector-like shaped areas corresponding to the convex vertexes minus these sector-like shaped areas corresponding to the concave vertexes is equal to $D$.
The sector-like shaped area corresponding to the concave vertex (x') is denoted as surplus-(x') in Fig. \ref{fixed_theta_multi}.

Similar to Fig. \ref{fixed_theta_overlap}-(a), ^^ ^^ surplus" is included in a congruent area of ^^ ^^ overlap" for each concave vertex.
Thus, the surplus area is smaller than the overlap area for each concave vertex.
Similar to Fig. \ref{fixed_theta_overlap}-(a), let $f(\xi_i,\phi_i|\theta)\geq 0$ be the overlap area size minus the surplus area size at the vertex of $\phi_i$ with fixed $\theta$. 
Similar to Fig. \ref{fixed_theta_overlap}-(b), $f(\xi_i,\phi_i|\theta)$ is a decreasing function of $\phi_i$.
Note that $f(\phi_i)\defeq \int_{-\pi}^{\pi} f(\xi_i,\phi_i|\theta) d\theta$ becomes independent of $\xi_i$ and is the same for any concave vertex because it is determined by $D$ and a given inner angle $\phi_i\leq \pi$.

Similar to Lemma \ref{lemma1}, $\int_{\route\cap D\intersect}\dD$ is $2\lengthx{D}L(\route)+2\pi \sizex{D}-\sum_i f(\phi_i)$, where $f(\phi_i)$ is a decreasing function of $\phi_i$.
The assumption that the distance between the $i$-th line segment and any point on the $j$-th line segment is larger than $\dia (D)$ for $j\neq i-1,i,i+1$ and for all $i$ is required.
This is because, under this assumption, $D$ in contact with the $i$-th line segment does not cause overlap with $D$ in contact with the $j$-th line segment where $j\neq i-1,i,i+1$.

\begin{figure}[htb] 
\begin{center} 
\includegraphics[width=8cm,clip]{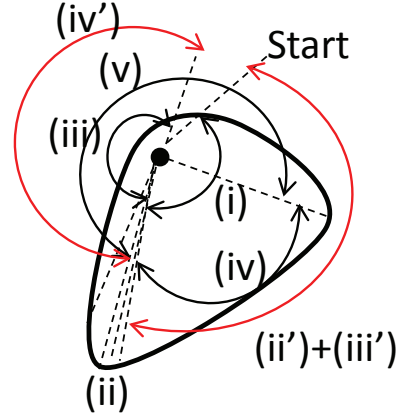} 
\caption{Parts of $D$ corresponding to sector-like shaped areas} 
\label{fixed_theta_angle} 
\end{center} 
\end{figure}

\end{proof}
When $n=1$, the term $\sum_i f(\phi_i)$ disappears.
As a result, $\int_{\route\cap D\intersect}\dD=2\lengthx{D}L(\route)+2\pi \sizex{D}$.
This result is consistent with Eq. (6.48) in \cite{Santalo}.

The following theorem based on Lemma \ref{lemma2} asserts that the zigzag route reduces $\Pr(\route\cap D\intersect)$.
If a route intersecting $D$ is disconnected due to a strong earthquake and there is a single route between $s$ and $t$, $\Pr(\route\cap D\intersect)$ is identical to $\disconnect$.
Therefore, under this assumption, the following theorem means that the zigzag route reduces $\disconnect$.
This is almost consistent with the result in \cite{infocom2014}, \cite{ToNsaito}, where a route with a shorter perimeter length of its convex-hull reduces $\disconnect$ when $D$ is modeled by a half-plane.

\begin{theorem}\label{single_theorem}
Assume that $\route$ and $\route'$ consist of $n$ line segments.
(One of two end points of each route is considered as a part of the 0-th and first line segments, and the other one is a part of the $n$-th and $(n+1)$-th line segments.)
Assume that the following assumption is satisfied:
The distance between the $i$-th line segment and any point on the $j$-th line segment is assumed to be larger than $\dia (D)$ for $j\neq i-1,i,i+1$ and for all $i$.

Let $\phi_i\leq \pi$ ($\phi_i'\leq \pi$) be the inner angle of the $i$-th and $(i+1)$-th line segments of $\route$ ($\route'$) for $i=1,2\cdots, n-1$.
Assume that $L(\route)\leq L(\route')$ and that $\{\phi_i\}_i\leq \{\phi'_j\}_j$.
(The meaning of $\{\phi_i\}_i\leq \{\phi'_j\}_j$ is as follows. When $\phi_{i_1}\leq \cdots\leq \phi_{i_{n-1}}$ and $\phi'_{j_1}\leq \cdots\leq \phi'_{j_{n-1}}$, $\phi_{i_k}\leq \phi'_{j_k}$ for $k=1,\cdots,n-1$.)
Then, $\Pr(\route\cap D\intersect)\leq \Pr(\route'\cap D\intersect)$.
\end{theorem}
\begin{proof}
Due to Lemma \ref{lemma2}, $\int_{\route\cap D\intersect}\dD=2\lengthx{D}L(\route)+2\pi \sizex{D}-\sum_i f(\phi_i)$ and $\int_{\route'\cap D\intersect}\dD=2\lengthx{D}L(\route')+2\pi \sizex{D}-\sum_i f(\phi'_i)$.
Because $f(\phi)$ is a decreasing function of $\phi$, $\sum_i f(\phi_i)=\sum_k f(\phi_{i_k})\geq \sum_k f(\phi'_{j_k})=\sum_i f(\phi'_i)$.
Therefore, $\int_{\route\cap D\intersect}\dD \leq\int_{\route'\cap D\intersect}\dD$.
\end{proof}

The theorem above provides a better route.
The following is such an example.
\begin{example}
In Fig. \ref{zigzagEx}, $\route$ consists of $l_i$ for $i=1$ to 4 and is the route in red (thin line), and $\route'$ consists of $l'_i$ for $i=1$ to 4 and is the route in black (thick line).
In Fig. \ref{zigzagEx}-(a), $\phi_2=\phi'_2$, $\phi'_1=\phi'_3=\pi$, $L(l_1)=L(l'_1)$, $L(l_4)=L(l'_4)$, $L(l_2)=L(l'_3)$, and $L(l_3)=L(l'_2)$.
In Fig. \ref{zigzagEx}-(b), $\phi_2=\phi'_2$, $\phi'_1=\phi'_3=\pi$, $L(l_1)\leq L(l'_1)$, $L(l_2)=L(l'_2)$, $L(l_3)=L(l'_3)$, and $L(l_4)\leq L(l'_4)$.
It is clear that $L(\route)\leq L(\route')$, $\{\phi_i\}_i\leq \{\phi'_j\}_j$ for both cases.
If $\dia (D)$ satisfies the assumptions of this theorem, $\Pr(\route\cap D\intersect)\leq \Pr(\route'\cap D\intersect)$ for both cases.

\begin{figure}[htb] 
\begin{center} 
\includegraphics[width=8cm,clip]{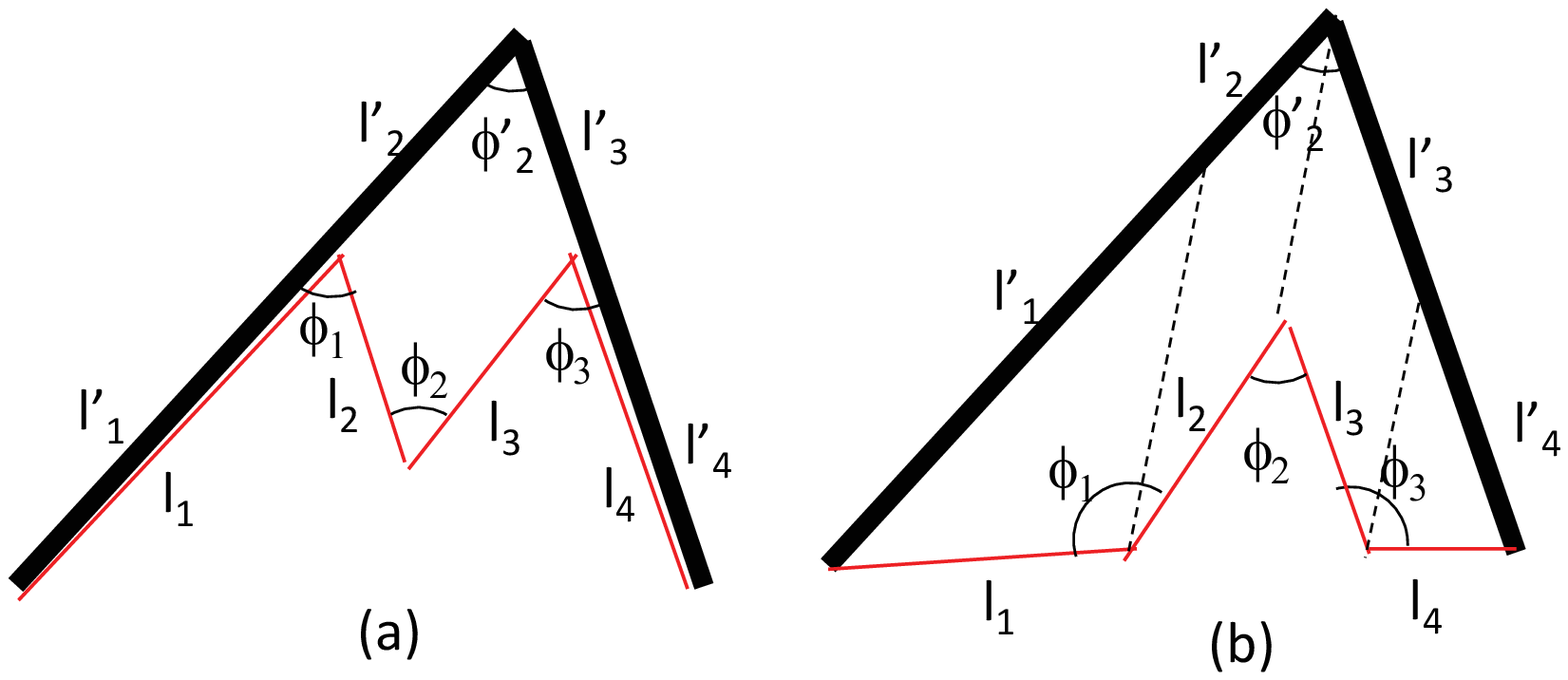} 
\caption{Example of $\Pr(\route\cap D\intersect)\leq \Pr(\route'\cap D\intersect)$} 
\label{zigzagEx} 
\end{center} 
\end{figure}

\end{example}

It is clear that $\Pr(\route\cap D\intersect)\leq \Pr(\bar{\route}\cap D\intersect)$, where $\bar{\route}$ is the convex hull of $\route$.
Because $\Pr(\route\cap D\intersect)=\Pr(\bar{\route}\cap D\intersect)$ in \cite{ToNsaito} assuming that $D$ is modeled by a half-plane, $\Pr(\bar{\route}\cap D\intersect)- \Pr(\route\cap D\intersect)$ affects finite $D$.

\subsubsection{Ring}
Let us consider a ring-type network.
\begin{theorem}\label{r_theo}
Assume that $\route_1$ and $\route_2$ form a ring-type network $\route_1\cup\route_2$, where $C(\route_1,\route_2)$ is convex.
If $\route_3, \route_4, \cdots\subset  C(\route_1,\route_2)$, $\Pr(\route_1\cap D\intersect, \route_2\cap D\intersect)\leq \Pr(\route_3\cap D\intersect, \route_4\cap D\intersect,\cdots)$.
\end{theorem}
\begin{proof}
If $\route_1\cap D\intersect, \route_2\cap D\intersect$, there exist point $u_i\subset D$ on $\route_i$ ($i=1,2$).
Due to the convexity of $D$ and $C(\route_1,\route_2)$, (i) any point on the line segment between $u_1$ and $u_2$ is in $D$, and (ii) one of the half planes made by the line passing $u_1$ and $u_2$ includes $s$ and the other half plane includes $t$.
Because of (ii), any route between $s$ and $t$ must intersect the line segment between $u_1$ and $u_2$.
Therefore, because of (i), any route includes a point in $D$.
This means that if $\{\route_1\cap D\intersect, \route_2\cap D\intersect\}$, then $\{\route_3,\route_4,\cdots\cap D\intersect\}$.
As a result, $\Pr(\route_1\cap D\intersect, \route_2\cap D\intersect)\leq \Pr(\route_3\cap D\intersect, \route_4\cap D\intersect,\cdots)$.


\end{proof}

\begin{cor}
Assume that $\route_1$ and $\route_2$ form a ring-type network $\route_1\cup\route_2$, where $C(\route_1,\route_2)$ is convex.
\begin{eqnarray}\label{ineq_cor}
&&\Pr(s\subset D)+\Pr(t\subset D)-\Pr(l(s,t)\subset D)\cr
&\leq&\Pr(\route_1\cap D\intersect, \route_2\cap D\intersect)\cr
&\leq& \Pr(l(s,t)\cap D\intersect)\cr
&=&(2\pi\sizex{D}+2d(s,t)\lengthx{D})/\sizex{\Omega(A_0)}
\end{eqnarray}
\end{cor}
\begin{proof}
The first inequality in Eq. (\ref{ineq_cor}) is because $\Pr(\route_1\cap D\intersect, \route_2\cap D\intersect)\geq \Pr(\{s \subset D\}\cup\{t \subset D\})=\Pr(s \subset D)+\Pr(t \subset D)-\Pr(s,t\subset D)$.
Due to the convexity of $D$, the event $\{s,t\subset D\}$ is equivalent to the event $\{l(s,t)\subset D\}$.
The second inequality in Eq. (\ref{ineq_cor}) is due to Theorem \ref{r_theo}.
The equality in Eq. (\ref{ineq_cor}) is because $\Pr(l(s,t)\cap D\intersect)=\sizex{\Omega(l(s,t))}/\sizex{\Omega(A_0)}$ due to the convexity of $l(s,t)$ and Eq. (6.48) in \cite{Santalo}.
Apply Eq. (\ref{size_two_convex}) and note that $\lengthx{l(s,t)}=2d(s,t)$.
\end{proof}
$\Pr(l(s,t)\subset D)$ is approximately given by Theorem \ref{include_theorem} in the Appendix as well as by Eqs. (6.44) and (6.52) in \cite{Santalo}.
Because of $\Pr(s \subset D)=\Pr(t \subset D)=2\pi\sizex{D}/\sizex{\Omega(A_0)}$, we can evaluate the first line of Eq. (\ref{ineq_cor}).

If a route intersecting $D$ is disconnected due to a strong earthquake, the probability $\Pr(\route_1\cap D\intersect, \route_2\cap D\intersect)$ is identical to $\disconnect$ for the ring-type network.
Therefore, Theorem \ref{r_theo} and the Corollary above provide the following result.

\begin{cor}
Assume that a route intersecting $D$ is disconnected.
If $\route_1$ and $\route_2$ form a ring-type network $\route_1\cup\route_2$, where $C(\route_1,\route_2)$ is convex, additional routes $\route_3, \route_4, \cdots\subset  C(\route_1,\route_2)$ do not decrease $\disconnect$.
If $\route_1$ and $\route_2$ are replaced with $\route_3, \route_4\subset  C(\route_1,\route_2)$, $\disconnect$ becomes worse.
When $\route_3, \route_4$ become $l(s,t)$,  $\disconnect$ becomes largest and is given by the right-hand side of Eq. (\ref{ineq_cor}).
$\disconnect$ cannot be smaller than the left-hand side of Eq. (\ref{ineq_cor}).
\end{cor}

The result that additional routes do not decrease $\disconnect$ under the assumption that routes intersecting $D$ are disconnected is consistent with the result in \cite{ToNsaito}, where $D$ is modeled by a half-plane.
However, the condition $\route_3, \route_4, \cdots\subset  C(\route_1,\route_2)$ is not needed in \cite{ToNsaito}.
On the other hand, the result that route $l(s,t)$ provides the worst $\disconnect$ is in clear contrast with the result in \cite{ToNsaito}.
When $D$ is modeled by a half-plane, the change in the physical routes of a ring-type network does not change $\disconnect$ \cite{ToNsaito}.

\subsection{Additive metrics}
In this subsection, the expected cost incurred is analyzed under the assumption that cost $w(u)$ is incurred when point $u$ on $\route$ is in $D$.

The assumption of convexity for $D$ is not required for Theorem \ref{add_theorem_single} and Corollaries \ref{add_cor_single}, \ref{add_cor_tree}, and \ref{N_d(s)}.

\begin{theorem}\label{add_theorem_single}
The expected cost incurred at the single $\route$ is given by the following.
\begin{equation}
E[\int_{u\subset \route,u\subset D}w(u)du]=2\pi \sizex{D}W_1/\sizex{\Omega(A_0)}\label{additive1},
\end{equation}
where $W_1\defeq \int_{u\subset \route}w(u)du$ is the sum of the costs along $\route$.

\end{theorem}
\begin{proof}
Note $E[\int_{u\subset \route, u\subset D}w(u)du]=\int_{u\subset D}(\int_{u\subset \route}w(u)du)\dD/\sizex{\Omega(A_0)}=\int_{u\subset \route}(\int_{u\subset D}\dD)w(u)du/\sizex{\Omega(A_0)}$.
Because $\int_{u\subset D}\dD=2\pi \sizex{D}$, $E[\int_{u\subset \route}w(u)du]=2\pi \sizex{D}\int_{u\subset \route}w(u)du/\sizex{\Omega(A_0)}$.
\end{proof}

This result shows that the expected cost is independent of the shape of $\route$ if $\int_{u\subset \route}w(u)du$ is fixed.
If $w(u)$ is constant, the expected cost is proportional to the route length.  
The extension to multiple routes is trivial.

\begin{cor}\label{add_cor_single}
The expected cost incurred at routes $\cup_i \route_i$ is given by the following.
\begin{equation}
E[\int_{u\subset \cup_i \route_i,u\subset D}w(u)du]=2\pi \sizex{D}W_1'/\sizex{\Omega(A_0)}\label{additive1},
\end{equation}
where $W_1'\defeq \int_{u\subset \cup_i \route_i}w(u)du$ is the sum of the costs along all routes $\route_1,\route_2,\cdots$.
\end{cor}

\subsubsection{Single route}
When there is a single $\route$ between $s$ and $t$, the following corollary can be provided.

\begin{cor}\label{add_cor_tree}
Assume that (i) a failure probability of the $i$-th node on $\route$ is $\alpha(i)$ and (ii) a link failure rate at $u$ is $\beta(u)$ per unit length, and (iii) the mean number $\bar\gamma$ of route failures defined by $\bar\gamma\defeq\sum_{i}\alpha(i)+L(\route) \bar\beta$ is much smaller than 1, where $\bar\beta\defeq \int_{u\subset \route}\beta(u)du/L(\route)$ is the mean link failure rate.
When there is a single $\route$ between $s$ and $t$, 
\begin{equation}
\disconnect=2\pi \sizex{D}\bar\gamma/\sizex{\Omega(A_0)}.\label{add_cor_tree_eq}
\end{equation}
\end{cor}
\begin{proof}
Assume that $\int_{u\subset \route}w(u)du\ll 1$, and each network component on $\route$ in $D$ is independently in failure.
Here, $w(u)du$ means the failure probability of $[u,u+du)$ when $[u,u+du)$ is in $D$.
Because the failure probability $w(u)du$ is very small and we can ignore multiple failures on $\route$, a single failure may occur with probability $\int_{u\subset \route, u\subset D}w(u)du$ for a fixed $D$.
Therefore, $\disconnect=E[\int_{u\subset \route,u\subset D}w(u)du]$.
Because of $W_1=\bar\gamma$ and Theorem \ref{add_theorem_single}, $\disconnect$ is given by Eq. (\ref{add_cor_tree_eq}).
\end{proof}

The result in \cite{ToNsaito} shows that $\disconnect$ can be improved by reducing the perimeter length of the convex hull of $\route$.
However, Corollary \ref{add_cor_tree} shows that reducing the route length is required to improve $\disconnect$ when $D$ is bounded.

\subsubsection{Ring}
For a ring-type network, the following theorem provides $\disconnect$.
The key observation is that the probability of more than two failures are negligibly small if $\sum_i\alpha(i)+\int_{u\subset \route_1\cup\route_2}\beta(u)du\ll 1$.
\begin{theorem}\label{ring_theorem}
Assume that $\route_1(s,t)$ and $\route_2(s,t)$ between $s$ and $t$ form a ring-type network $\route_1\cup\route_2$.
When $\sum_i\alpha(i)+\int_{u\subset \route_1\cup\route_2}\beta(u)du\ll 1$, 
\begin{eqnarray}
&&\disconnect\cr
&=&\alpha(s)\Pr(s\subset D)+\alpha(t)\Pr(t\subset D)\cr
&&-\alpha(s)\alpha(t)\Pr(l(s,t)\subset D)\cr
&&+\sum_{i(\neq s,t)\subset\route_1,j(\neq s,t)\subset\route_2}\alpha(i)\alpha(j)\Pr(l(i,j)\subset D)\cr
&&+\sum_{k=1}^2\sum_{v(\neq s,t)\subset\route_k}\alpha(v)\sum_{\Route(i)\subset \route_j(\neq \route_k)} W_i(v)\cr
&&+\sum_{\Route(i)\subset\route_1,\Route(j)\subset\route_2}W_{i,j}\label{ring_theorem_eq},
\end{eqnarray}
where $ W_i(v)\defeq\int_{u\subset \Route(i)}\beta(u)\Pr(l(u,v)\subset D)du$, $W_{i,j}\defeq\int_{u\subset \Route(i), v\subset \Route(j)}\beta(u)\beta(u)\Pr(l(u,v)\subset D)du\,dv$, and $\Route(i)$ means the part between node-$i$ and node-$(i+1)$ on $\route_1\cup\route_2$.
\end{theorem}
\begin{proof}
When $\route_1(s,t)$ and $\route_2(s,t)$ form a ring-type network, $s$ and $t$ are disconnected if the event in which $s$ or $t$ is in failure occurs or if the event $u_1(\neq s,t)$ on $\route_1(s,t)$ in $D$ and $u_2(\neq s,t)$ on $\route_2(s,t)$ in $D$ are in failure occurs.
Let $P_1$ and $P_2$ be the probability of the occurrence of the former event and the latter event, respectively.
Because the probability that both events occur is much smaller than the probabilities that one of these event occurs, we can ignore the probability that both events occur.
In addition, due to the convexity of $D$, the event $\{u_1,u_2\subset D\}$ is equivalent to the event $\{l(u_1,u_2)\subset D\}$.
Therefore, $\disconnect=P_1+P_2$ and $P_1, P_2$ are given below. 
\begin{eqnarray*}
P_1&=&\alpha(s)\Pr(s\subset D)+\alpha(t)\Pr(t\subset D)\cr
&&-\alpha(s)\alpha(t)\Pr(l(s,t)\subset D)
\end{eqnarray*}
\begin{eqnarray*}
P_2=\int_{u\subset \route_1(s,t), v\subset \route_2(s,t)}\gamma(u)\gamma(v)\Pr(l(u,v)\subset D)du\,dv
\end{eqnarray*}
where $\gamma(x)\defeq\sum_{i(\neq s,t)}\alpha(i)\bfone(i=x)+\beta(x)$, which is a failure rate at $x$.
Therefore, 
\begin{eqnarray*}
&&P_2\cr
&=&\sum_{i(\neq s,t)\subset\route_1,j(\neq s,t)\subset\route_2}\alpha(i)\alpha(j)\Pr(l(i,j)\subset D)\cr
&&+\sum_{k=1}^2\sum_{i(\neq s,t)\subset\route_k}\alpha(i)\int_{u\subset \route_j(\neq \route_k)}\beta(u)\Pr(l(u,i)\subset D)du\cr
&&+\int_{u\subset \route_1(s,t), v\subset \route_2(s,t)}\beta(u)\beta(u)\Pr(l(u,v)\subset D)du\,dv.
\end{eqnarray*}
Because $\int_{u\subset \route_k}\beta(u)\Pr(l(u,i)\subset D)du=\sum_{\Route(j)\subset \route_k} W_j(i)$ and $\int_{u\subset \route_1(s,t), v\subset \route_2(s,t)}\beta(u)\beta(u)\Pr(l(u,v)\subset D)du\,dv=\sum_{\Route(i)\subset\route_1,\Route(j)\subset\route_2}W_{i,j}$, Theorem \ref{ring_theorem} is derived.
\end{proof}

When $\beta(u)=\beta_j$ for $u\subset\Route(j)$ for all $j$,
\begin{eqnarray*}
W_{i,j}&=&\beta_i\beta_j\int_{u\subset \Route(i), v\subset \Route(j)}\Pr(l(u,v)\subset D)du\,dv.
\end{eqnarray*}
Because
\begin{eqnarray*}
&&\int_{u\subset \Route(i), v\subset \Route(j)}\Pr(l(u,v)\subset D)du\,dv\cr
&\approx& \frac{1}{2}\int_{u\subset \Route(i), v\subset \Route(j)}(\Pr(l(u,j)\subset D)\cr
&&\qquad\qquad+\Pr(l(u,j+1)\subset D))du\,dv\cr
&=&\frac{1}{2}L(\Route(j))\int_{u\subset \Route(i)}(\Pr(l(u,j)\subset D)\cr
&&\qquad\qquad+\Pr(l(u,j+1)\subset D))du\cr
&\approx&\frac{1}{4}L(\Route(j))\int_{u\subset \Route(i)}\sum_{k_1=i}^{i+1}\sum_{k_2=j}^{j+1}\Pr(l(k_1,k_2)\subset D)du\cr
&=&\frac{1}{4}L(\Route(i))L(\Route(j))\sum_{k_1=i}^{i+1}\sum_{k_2=j}^{j+1}\Pr(l(k_1,k_2)\subset D),
\end{eqnarray*}
$W_{i,j}\approx\sum_{k_1=i}^{i+1}\sum_{k_2=j}^{j+1}\frac{\beta_i\beta_j}{4}L(\Route(i))L(\Route(j))\Pr(l(k_1,k_2)\subset D)$.
Similarly, $ W_i(v)\approx\beta_i L(R(i))(\Pr(l(v,i)\subset D)+\Pr(l(v,i+1)\subset D))/2$.
By applying these approximations to Eq. (\ref{ring_theorem_eq}), the following corollary is obtained.

\begin{cor}\label{ring_cor}
Assume that $\route_1(s,t)$ and $\route_2(s,t)$ between $s$ and $t$ form a ring-type network.
When $\sum_{i\subset \route_1\cup\route_2}\alpha(i)+\bar\beta L(\route_1\cup\route_2)\ll 1$, and $\beta(u)=\beta_j$ for $u\subset\Route(j)$,
\begin{eqnarray}
&&\disconnect\cr
&\approx&2\pi\sizex{D}(\alpha(s)+\alpha(t))/\sizex{\Omega(A_0)}\cr
&&-\alpha(s)\alpha(t)\Pr(l(s,t)\subset D)\cr
&&+\sum_{i(\neq s,t)\subset\route_1,j(\neq s,t)\subset\route_2}\alpha(i)\alpha(j)\Pr(l(i,j)\subset D)\cr
&&+\sum_{k=1}^2\sum_{v(\neq s,t)\subset\route_k}\alpha(v)\sum_{\Route(i)\subset \route_j(\neq\route_k)}\beta_i L(R(i))\cr
&&\qquad\qquad(\Pr(l(v,i)\subset D)+\Pr(l(v,i+1)\subset D))/2\cr
&&+\sum_{\Route(i)\subset\route_1,\Route(j)\subset\route_2}\sum_{k_1=i}^{i+1}\sum_{k_2=j}^{j+1}\frac{\beta_i\beta_j}{4}L(\Route(i))L(\Route(j))\cr
&&\qquad\qquad\qquad\qquad\qquad\qquad
\Pr(l(k_1,k_2)\subset D).\label{ring_cor_eq}
\end{eqnarray}
\end{cor}
Because $\Pr(l(u,v)\subset D)$ can be evaluated using Theorem \ref{include_theorem} in the Appendix as well as with Eqs. (6.44) and (6.52) in \cite{Santalo}, the equation above can be evaluated.



\subsection{Optimization}
The results mentioned above provide network optimization.
First, define a metric. The following is an example.
\begin{cor}\label{N_d(s)}
Assume that the mean number $\bar\gamma_i$ of route failures for $\route_j$ satisfies $\bar\gamma_i\ll 1$.
When there is a single route $\route_j$ between $s$ and $t_j$ ($j=1,2\cdots$), the mean number $N_d(s)$ of nodes $t_j$ disconnecting node $s$ is given by 
\begin{equation}
N_d(s)=2\pi \sizex{D}\sum_j\bar\gamma_j/\sizex{\Omega(A_0)}.
\end{equation}
\end{cor}
This is because $N_d(s)=\sum_j E[\bfone(s\not \leftrightarrow t_j)]=\sum_j \Pr(s\not \leftrightarrow t_j)$.

Assume that any communication needs to visit a server $S_v$.
When this server is placed at $s^*$, the mean number of disconnected communications due to server visit failure is minimized.
Here, $s^*={\rm argmin}_s N_d(s)$ for a tree network.
When this server is placed at $s^\dagger$, the worst disconnect probability due to server visit failure is minimized.
Here, $s^\dagger={\rm argmin}_s \max_t \disconnect$.

A similar metric and similar optimization for a ring network can be defined.


\section{Numerical examples}
\subsection{Disaster area data}
This paper uses maps showing past earthquake intensities stronger than 5- on the Japanese scale of 7 as $D$.
The maps are released by Japan Meteorological Agency (JMA) \cite{JMA}.
However, an earthquake of which the seismic center is in the sea was not used because only onshore parts of its intensity map are available.
That is, the whole shape of $D$ cannot be obtained if its seismic center is in the sea.
Therefore, the data regarding the earthquake on March 11, 2011, which caused the largest amount of damage in Japan after World War II, was not used.
As a result, eight samples of $D$ were obtained \cite{JMA_data}.
In addition, the earthquake that occurred in 1995 and caused the second largest amount of damage in Japan was used.
Because the maps released by JMA cover only earthquakes that occurred over these several years, another map was used \cite{hanshin}.
The map is that of intensity 7 area. 
In the remainder of this section, nine earthquakes are numbered in descending order of $\sizex{D}$.
This figure shows that actual $D$ may not be convex or consist of a single part.
Many earthquakes are similar to Earthquake 1, but Earthquakes 3 and 9 are different.
Because of their shapes, their perimeters are large compared with their sizes.



Because the convexity of $D$ is assumed with most of the results mentioned above, the numerical examples can be used to confirm that these results are actually valid.

In the following numerical examples, simulation, as well as equations derived in this paper, was used.
In the simulation, the reference point $(x_d,y_d)$ is randomly and uniformly distributed in an area including $A_0$, and the direction $\theta$ is randomly and uniformly distributed in $[0,2\pi)$.
10,000 samples was used for each point in a graph.

\subsection{Real network}
To evaluate the accuracy of the approximation formula (Eq. \ref{ring_cor_eq}) and the tightness of the upper and lower bounds (Eq. \ref{ineq_cor}) for a real physical network, the network shown in Fig. \ref{route} was used.
A disk with a 10-km radius including the network was used as $A_0$.

\begin{figure}[htb] 
\begin{center} 
\includegraphics[width=8cm,clip]{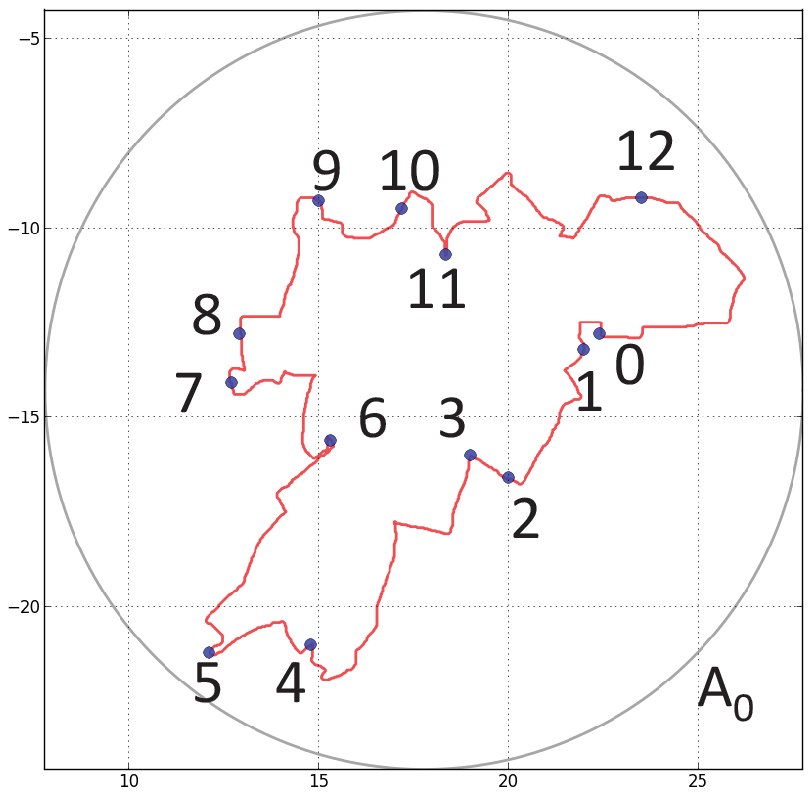} 
\caption{Example of $\route$} 
\label{route} 
\end{center} 
\end{figure}

In this network, the distance between node 5 and node 12 is the largest among two nodes.
Figure \ref{tokyo5-12Hisai} plots $\Pr(\route_1\cap D\intersect, \route_2\cap D\intersect)$ and Figure \ref{tokyo5-12} plots $\disconnect$ between these two nodes.

\begin{figure}[htb] 
\begin{center} 
\includegraphics[width=8cm,clip]{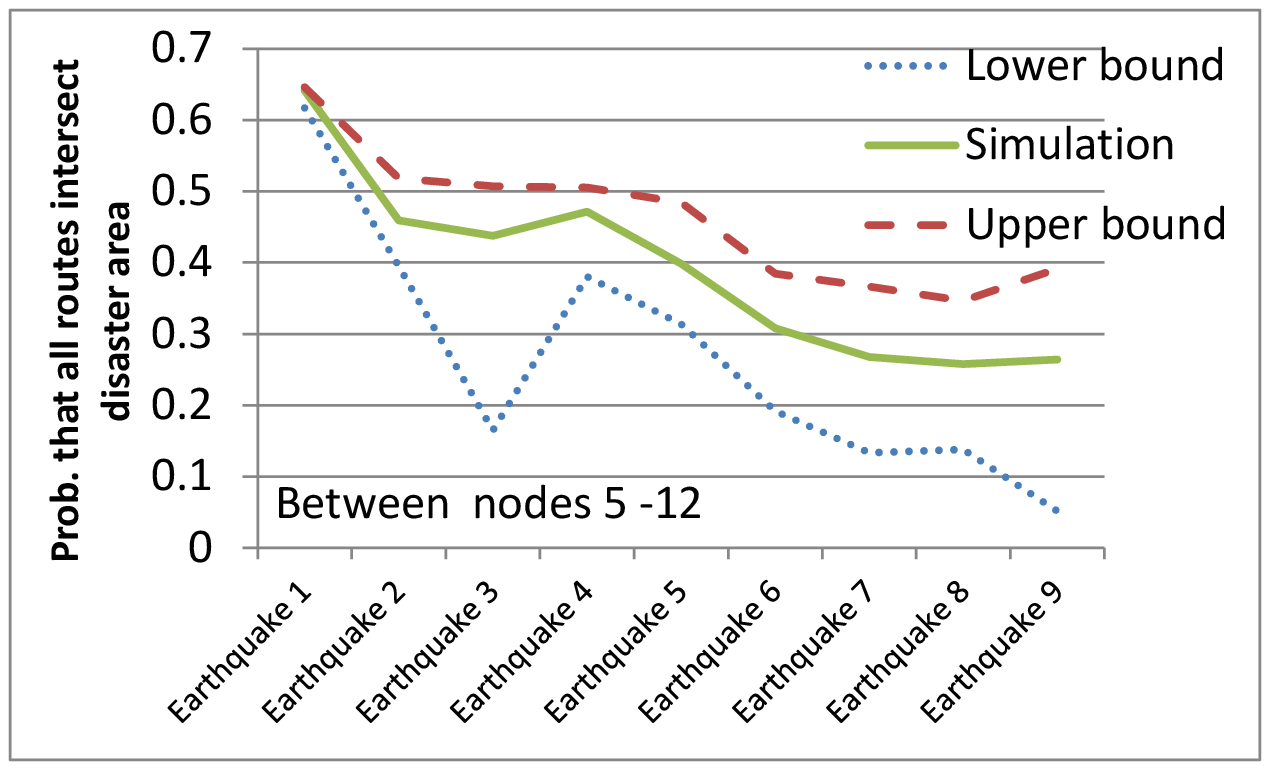} 
\caption{$\Pr(\route_1\cap D\intersect, \route_2\cap D\intersect)$ for real network} 
\label{tokyo5-12Hisai} 
\end{center} 
\end{figure}

\begin{figure}[htb] 
\begin{center} 
\includegraphics[width=8cm,clip]{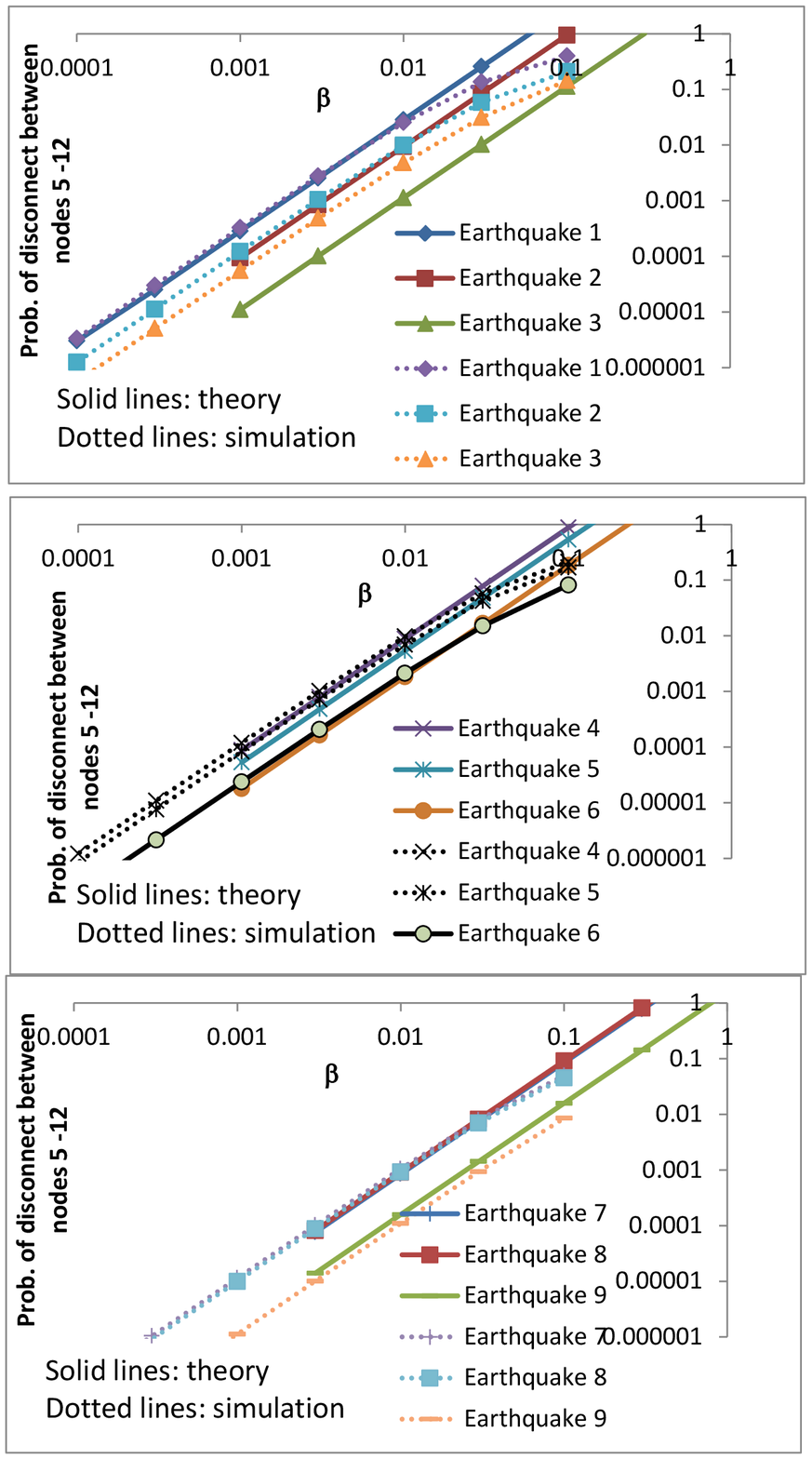} 
\caption{$\disconnect$ for real network} 
\label{tokyo5-12} 
\end{center} 
\end{figure}

In Fig. \ref{tokyo5-12Hisai}, $\Pr(\route_1\cap D\intersect, \route_2\cap D\intersect)$ derived using the simulation is almost within the upper and lower bounds, although the network is not convex.
Precisely, $\Pr(\route_1\cap D\intersect, \route_2\cap D\intersect)$ for Earthquake 1 is slightly larger than the upper bound.
The difference between the two bounds is normally larger when the earthquake is smaller, but it is not so simple.
Figure \ref{tokyo5-12} shows that $\disconnect$ has good agreement with the theoretical one for most of the earthquakes except for Earthquake 3, which has a complicated boundary.
For this earthquake, the theory underestimates $\disconnect$ by several times.

\subsection{Single route}
Consider the simple example networks in Figure \ref{num_ex_net1}.
Note that these networks have the same route length $L(\route)$.

According to Theorem \ref{single_theorem}, assuming convexity of $D$, $\Pr(\route \cap D\intersect)$ becomes smaller as $n_{fold}$ increases.
The simulation was conducted to verify this theorem for non-convex $D$s.
The results are shown in Fig. \ref{single_result1}.
It illustrates that this theorem is valid for actual earthquakes.
$\Pr(\route \cap D\intersect)$ decreased by about 10\% or less when $n_{fold}=7$ was used instead of $n_{fold}=1$.
It also shows that $\Pr(\route \cap D\intersect)$ often decreased as the size of $D$ decreased.
However, this is not always the case.  
(Although the earthquakes are numbered in descending order of size, $\Pr(\route \cap D\intersect)$ is not the order. That is, it depends on the shape of $D$.)
Its dependence depends on the network size or the route length.

\begin{figure}[htb] 
\begin{center} 
\includegraphics[width=8cm,clip]{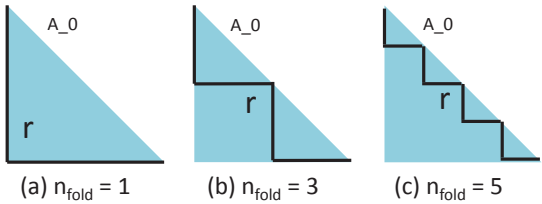} 
\caption{Simple example networks} 
\label{num_ex_net1} 
\end{center} 
\end{figure}

\begin{figure}[htb] 
\begin{center} 
\includegraphics[width=8cm,clip]{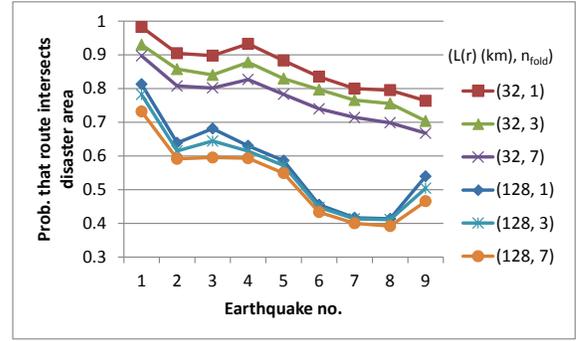} 
\caption{Verification of Theorem \ref{single_theorem}} 
\label{single_result1} 
\end{center} 
\end{figure}

For the same example networks, the probability of disconnection is evaluated through Corollary \ref{add_cor_tree} as well as simulation.
Figure \ref{single_result2} plots the result for Earthquake 1 as an example.
Here, ^^ ^^ Theory" is $\disconnect$ derived by Corollary \ref{add_cor_tree} and is independent of $n_{fold}$.
The other curves are derived by simulation.
Among them, ^^ ^^ Upper" means $\Pr(\route \cap D\intersect)$ obtained by simulation.

$\disconnect$ approaches $\Pr(\route \cap D\intersect)$, as $\beta$ increases because $\route \cap D\intersect$ means $s \not \leftrightarrow t$ when $\beta$ is large.
In addition, as expected, ^^ ^^ Theory" has good agreement with the simulation results when $\beta$ is small.

\begin{figure}[htb] 
\begin{center} 
\includegraphics[width=8cm,clip]{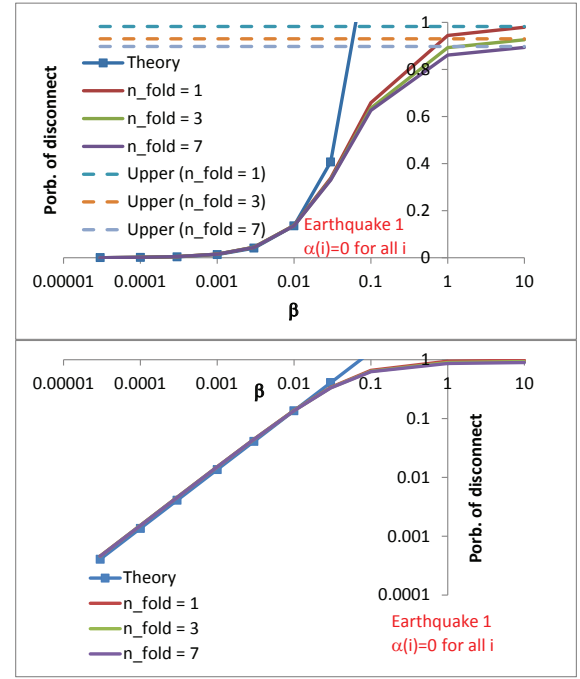} 
\caption{Verification of Corollary \ref{add_cor_tree}} 
\label{single_result2} 
\end{center} 
\end{figure}

\subsection{Ring network}
Consider the simple example networks in Figure \ref{num_ex_net2}-(1).

According to Theorem \ref{r_theo}, $\Pr(\route_1 \cap D\intersect,\route_2 \cap D\intersect)$ for a circle (Fig. \ref{num_ex_net2}-(1-a)) is the smallest, that for a hexagon (Fig. \ref{num_ex_net2}-(1-b)) is the second smallest, and that for a concave (Fig. \ref{num_ex_net2}-(1-c)) is the largest if $D$ is convex.
A simulation was conducted to verify this theorem for non-convex $D$s.

\begin{figure}[htb] 
\begin{center} 
\includegraphics[width=8cm,clip]{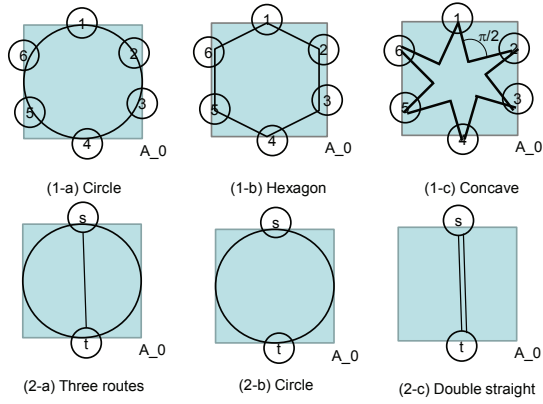} 
\caption{Simple example network2} 
\label{num_ex_net2} 
\end{center} 
\end{figure}

The results for Earthquakes 1 and 9 are shown in Fig. \ref{ring_result1}.
Although it is difficult to distinguish $\Pr(\route_1\cap D\intersect, \route_2\cap D\intersect)$ for a circle and that for a hexagon, the results seem consistent with Theorem \ref{r_theo}.
As the number of hops between $s$ and $t$ increases, $\Pr(\route_1\cap D\intersect, \route_2\cap D\intersect)$ increases.
This is because the probability of disconnection for the shorter of the two routes $\route_1, \route_2$ is dominant in $\Pr(\route_1\cap D\intersect, \route_2\cap D\intersect)$.

\begin{figure}[htb] 
\begin{center} 
\includegraphics[width=8cm,clip]{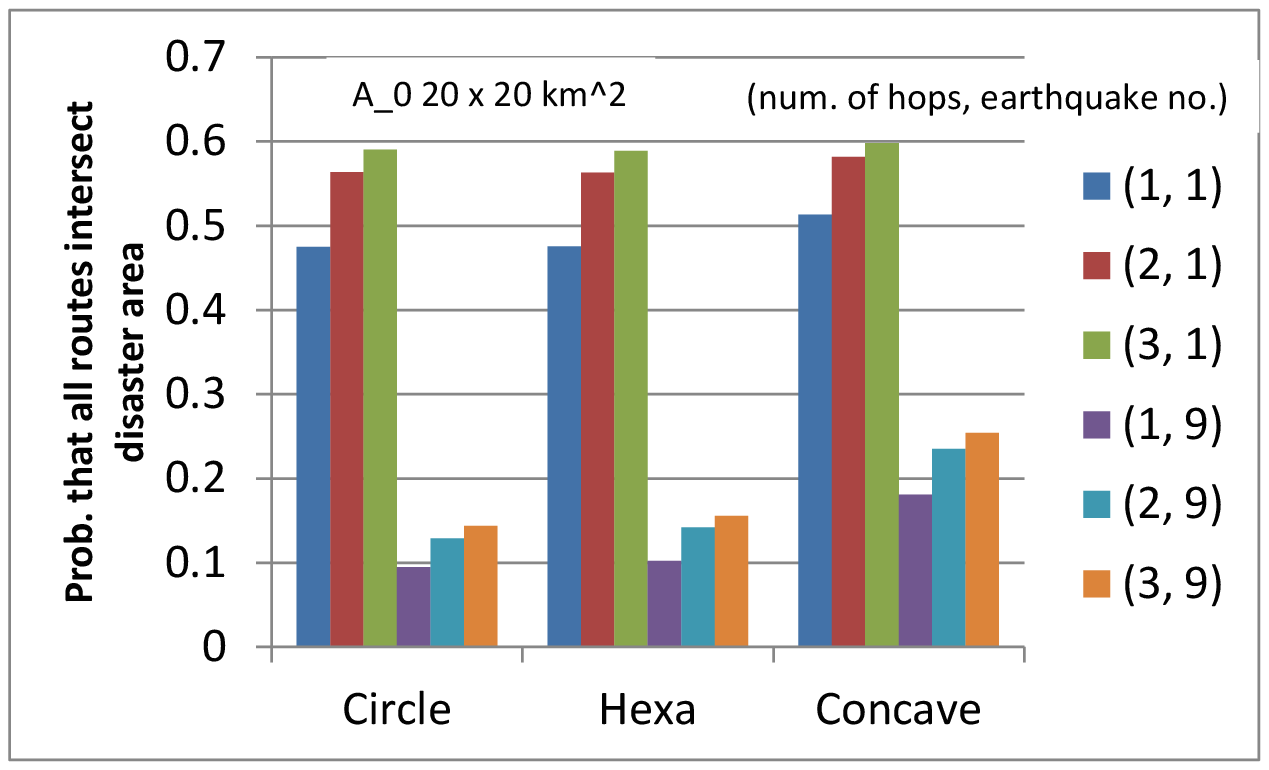} 
\caption{$\Pr(\route_1\cap D\intersect, \route_2\cap D\intersect)$ for Earthquake 1} 
\label{ring_result1} 
\end{center} 
\end{figure}

Figure \ref{ring_result2} evaluates $\disconnect$ through simulation and Corollary \ref{ring_cor}.
As expected, ^^ ^^ Theory" has good agreement with the simulation results when $\beta$ is small (Fig. \ref{ring_result2}-(b)).
Unfortunately, when the network becomes large (or $D$ becomes small), the accuracy of the theoretical result deteriorates (Fig. \ref{ring_result2}-(c)).
This is because the approximation error of Corollary \ref{ring_cor} increases.

\begin{figure}[htb] 
\begin{center} 
\includegraphics[width=8cm,clip]{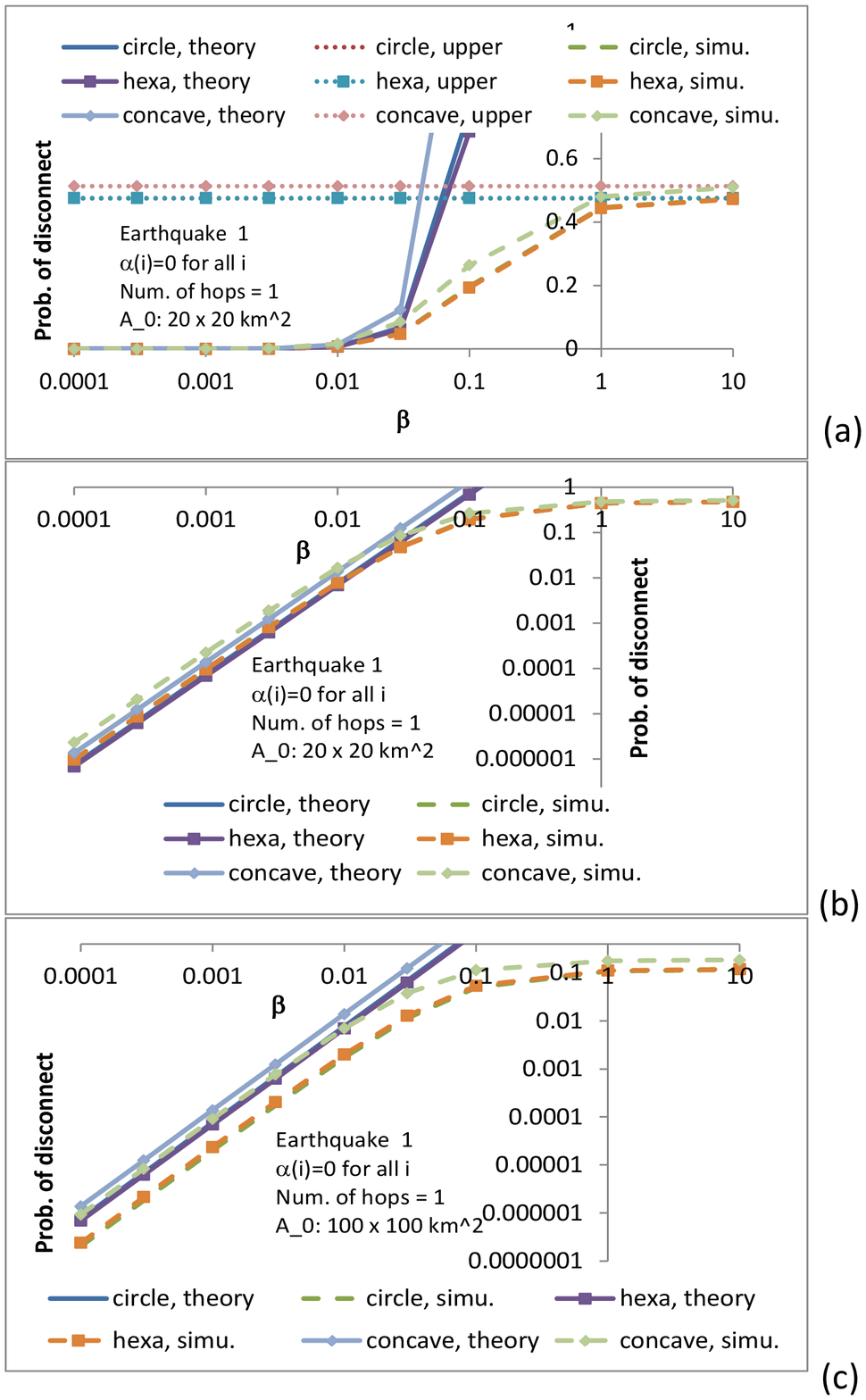} 
\caption{$\disconnect$ for Earthquake 1} 
\label{ring_result2} 
\end{center} 
\end{figure}

Consider the simple example networks in Figure \ref{num_ex_net2}-(2).
To verify Theorem \ref{r_theo}, a simulation was conducted.
Figure \ref{ring_result3} plots the probability that $\route_i\cap D\intersect$ for every $\route_i$.
As Theorem \label{r_intersect} asserts, this probability for ^^ ^^ Circle" is not improved by ^^ ^^ Three routes" and is better than that for ^^ ^^ Straight" for many earthquakes such as Earthquake 1.
In addition, when the network is small (or $D$ is large), the probabilities for ^^ ^^ Three routes", ^^ ^^ Circle," and ^^ ^^ Straight" become identical \cite{ToNsaito}.
However, there are exceptions.
Because Earthquake 9 consists of multiple separated subareas, Theorem \ref{r_theo} is not valid.
Here, ^^ ^^ Three routes" is the best and ^^ ^^ Straight" is the worst.

\begin{figure}[htb] 
\begin{center} 
\includegraphics[width=8cm,clip]{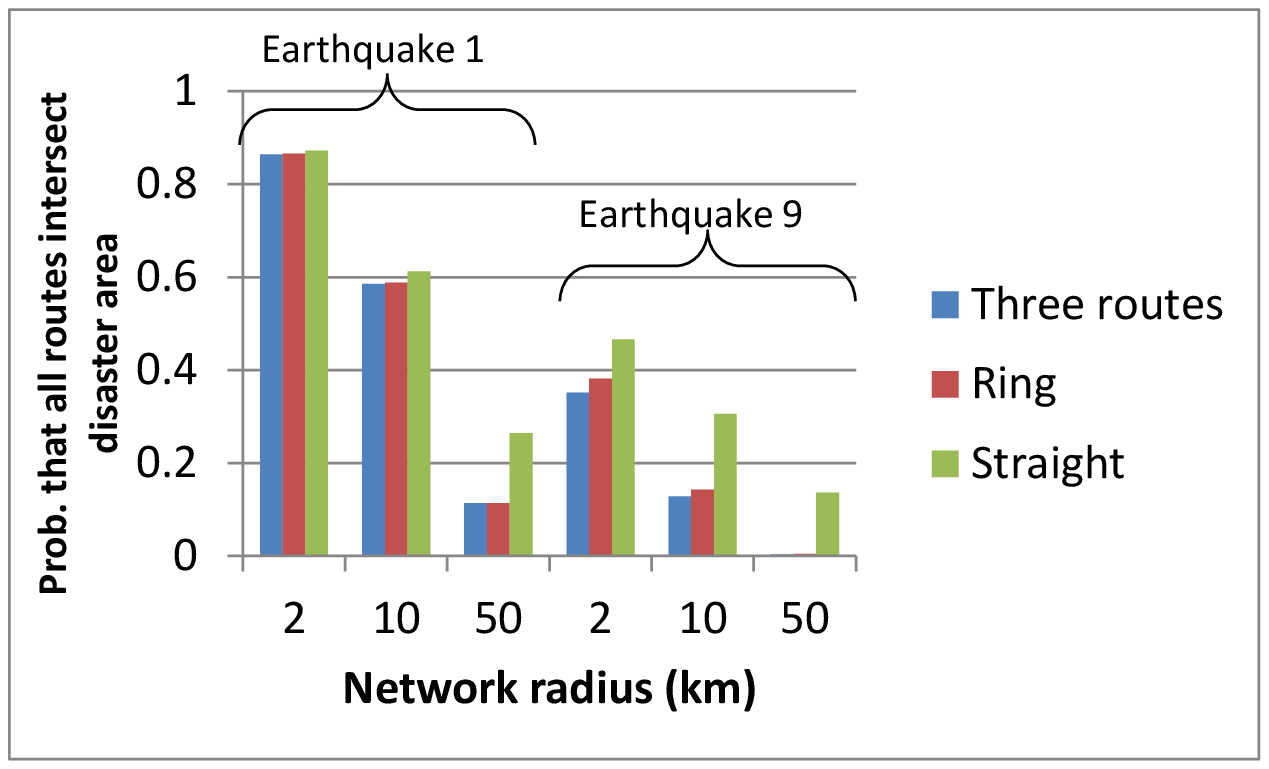} 
\caption{$\Pr(\forall i, \route_i\cap D\intersect)$} 
\label{ring_result3} 
\end{center} 
\end{figure}

\subsection{Optimization}
Let us investigate where a server should be placed in the network shown in Fig. \ref{route}.
Assume that there is a weak link in this network and that the link failure rate $\beta^\dagger$ of the weak link is larger than the link failure rate $\beta$ of other links.
The optimal location is determined by minimizing the worst disconnect probability due to server visit failure.
That is, the server should be located at $s^\dagger={\rm argmin}_s \max_t \disconnect$.

Figure \ref{weak} plots the result.
This figure demonstrates that $s^\dagger$ is far from the weak link, although a similar result is shown in \cite{infocom2014}.
However, this result depends on $\beta^\dagger$.
If the difference between $\beta^\dagger$ and $\beta$ is small, it is not clear that $s^\dagger$ is far from the weak link.
This trend becomes clear as the difference increases.

\begin{figure}[htb] 
\begin{center} 
\includegraphics[width=8cm,clip]{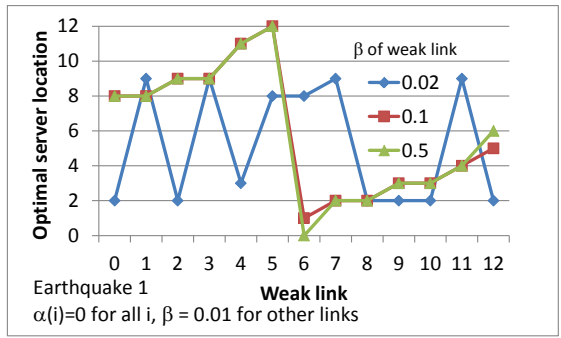} 
\caption{Optimal server location} 
\label{weak} 
\end{center} 
\end{figure}

\section{Conclusion}
The physical network route shape against a generic and bounded disaster area was analyzed.
Based on the analysis results, the following design rules were proposed: (i) a shorter zigzag route is appropriate to reduce the probability that the route intersects the disaster area, (ii) the route length should be reduced for a single route if an additive performance metric needs to be reduced because that metric is independent of the network shape for fixed route length, (iii) additional routes within a ring network are useless to decrease the probability that routes intersect the disaster area, (iv) a wider detour route for a ring network should be adopted to reduce that probability.
These analysis results as well as formulas evaluating the probability of disconnecting two given nodes were validated through empirical earthquake data.

The analysis results are also useful for spatial design of network elements rather than the physical network shape.
An optimal server placement was discussed as an example.

Although existing disaster management is based on protection and restoration, the proposed design method is the first step in disaster management aiming at disaster avoidance.  

\begin{biography}
{Hiroshi Saito} graduated from the University of Tokyo with a B.E. degree in Mathematical Engineering in 1981, an M.E. degree in Control Engineering in 1983, and a Dr.Eng. in Teletraffic Engineering in 1992.

He joined NTT in 1983. He is currently an Executive Research Engineer at NTT Network Technology Labs. He received the Young Engineer Award of the Institute of Electronics, Information and Communication Engineers (IEICE) in 1990, the Telecommunication Advancement Institute Award in 1995 and 2010, and the excellent papers award of the Operations Research Society of Japan (ORSJ) in 1998. He served as an editor and a guest editor of technical journals such as Performance Evaluation, IEEE Journal of Selected Areas in Communications, and IEICE Trans. Communications.
He was the director of Journal and Transactions of IEICE, the organizing committee/program committee chairman of a few international conferences, and a program committee member of more than 30 international conferences. He is currently an editorial board member of Computer Networks. Dr. Saito is a fellow of IEEE, IEICE, and ORSJ, and a member of IFIP WG 7.3. His research interests include traffic technologies of communications systems, network architecture, and ubiquitous systems.

More information can be found at http://www9.plala.or.jp/hslab.
\end{biography}

\appendices
\section{Approximation formula for $\Pr(l(u,v)\subset D)$}
\subsection{Proposed approximation formula}
Under restricted conditions, $\Pr(l(u,v)\subset D)$ is given by Eqs. (6.44) and (6.52) in \cite{Santalo}.
However, the numerical examples using the nine earthquakes do not often satisfy these conditions.
As a result, large errors appear.
Therefore, an approximation formula for $\Pr(l(u,v)\subset D)$ is proposed.

\begin{lemma}\label{disk_approx}
When $D$ is a disk with radius $R_D$, $\Pr(l(u,v)\subset D)\approx g_1(d(u,v),R_D)/\sizex{\Omega(A_0)}$
Here,
\begin{eqnarray}
&&g_1(d(u,v),R_D)\cr
&\defeq& \cases{4\pi\{R_D^2\arccos(\frac{d(u,v)}{2R_D})\cr
 -\frac{d(u,v)R_D}{2}\sqrt{1-(\frac{d(u,v)}{2R_D})^2}\},&for $d(u,v)<2R_D$,\cr 0, &otherwise.}\cr
&&
\end{eqnarray}
\end{lemma}
\begin{proof}
The center of $l(u,v)$ can be located in the part surrounded by the red dotted curves in Fig. \ref{include_disk}.
Its size is given by $g_1(d(u,v),R_D)$.

\begin{figure}[htb] 
\begin{center} 
\includegraphics[width=8cm,clip]{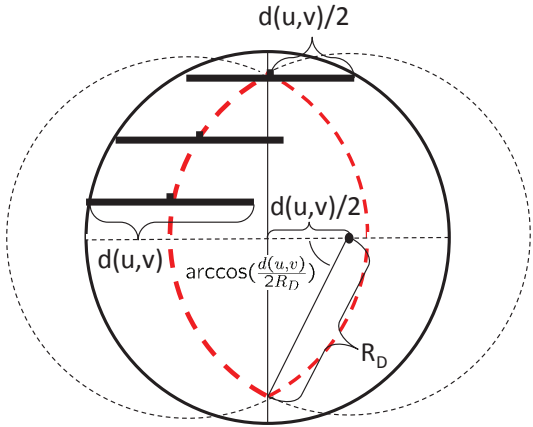} 
\caption{$\Pr(l(u,v)\subset D)$ for disk-shaped $D$} 
\label{include_disk} 
\end{center} 
\end{figure}

\end{proof}

\begin{lemma}\label{rect_approx}
When $D$ is a rectangle with edge lengths $a_D\leq b_D$, $\Pr(l(u,v)\subset D)\approx g_2(d(u,v),a_D,b_D)/\sizex{\Omega(A_0)}$.
For $d(u,v)\leq a_D$,
$$g_2(d(u,v),a_D,b_D)\defeq 2\pi a_D{ b_D}+2d(u,v)^2-4d(u,v)( a_D+ b_D), $$
for $a_D\leq d(u,v) \leq b_D$, 
\begin{eqnarray*}
&&g_2(d(u,v),a_D,b_D)\cr
&\defeq& 4a_Db_D(\pi/2-\arccos(\frac{a_D}{d(u,v)}))-4d(u,v)b_D\cr
&&+4b_D\sqrt{d(u,v)^2-a_D^2}-2a_D^2,
\end{eqnarray*}
for $b_D\leq d(u,v) \leq\sqrt{ a_D^2+b_D^2}$, 
\begin{eqnarray*}
&&g_2(d(u,v),a_D,b_D)\cr
&\defeq& 4a_D{ b_D}(\pi/2-\arccos(\frac{{ b_D}}{d(u,v)})-\arccos(\frac{{ a_D}}{d(u,v)}))\cr
&&+4{ a_D}\sqrt{d(u,v)^2-{ b_D}^2}+4{ b_D}\sqrt{d(u,v)^2-{ a_D}^2}\cr
&&-2{ b_D}^2-2{ a_D}^2-2d(u,v)^2, 
\end{eqnarray*}
and for $\sqrt{a_D^2+b_D^2} <d(u,v)$, $$g_2(d(u,v),a_D,b_D)\defeq 0.$$
\end{lemma}
The above lemma is given by Eq. (13) in \cite{mobileComp}.
(Because the direction of $l(u,v)$ is defined in $[0,2\pi)$ in this paper and in $[0,\pi)$ in \cite{mobileComp}, $g_2(d(u,v),a_D,b_D)$ doubles in Eq. (13) in \cite{mobileComp}.)

\begin{theorem}\label{include_theorem}
If $\lengthx{D}^2\geq 16\sizex{D}$, $\Pr(l(u,v)\subset D)\approx g_2(d(u,v),a_D,b_D)/\sizex{\Omega(A_0)}$, where $a_D=(\lengthx{D}-\sqrt{\lengthx{D}^2-16\sizex{D}})/4$, $b_D=(\lengthx{D}+\sqrt{\lengthx{D}^2-16\sizex{D}})/4$.
Otherwise, $\Pr(l(u,v)\subset D)\approx g_1(d(u,v),\frac{\lengthx{D}}{2\pi})/\sizex{\Omega(A_0)}$
\end{theorem}
\begin{proof}
If $\lengthx{D}^2\geq 16\sizex{D}$, there exist $a_D$ and $b_D$ satisfying $\sizex{D}=a_Db_D$ and $\lengthx{D}=2(a_D+b_D)$.
Then, apply Lemma \ref{rect_approx}.
Otherwise, apply Lemma \ref{disk_approx} with $R_D=\frac{\lengthx{D}}{2\pi}$.

\end{proof}
There is another option that $R_D=\sqrt{\sizex{D}/\pi}$ in the approximation formula above.

\subsection{Numerical examples}
The following numerical examples are evaluated to compare the accuracy of the proposed approximation (Theorem \ref{include_theorem}) and that of Eqs. (6.44) and (6.52) in \cite{Santalo}.
Because Eq. (6.44) requires the angle of each vertex of $D$, assume that the angles of all the vertexes are the same and that the number of vertexes is 3, 4, 6, or 12.

Nine earthquakes are used as $D$.
The line segment length $d(u,v)(=L(l(u,v)))$ is $\dia(D)\times$0.2, 0.4, 0.6, or 0.8, and $A_0$ is a square with an edge length $2\dia(D)$.

Figure \ref{prob_include} plots the median of the normalized $\Pr(l(u,v)\subset D)$, that is, $\Pr(l(u,v)\subset D)$ of each approximation divided by that of the simulation, over nine earthquakes.

This figure shows that the proposed approximation clearly outperforms Eqs. (6.44) and (6.52) in \cite{Santalo}.
In particular, the proposed approximation provides positive $\Pr(l(u,v)\subset D)$, although Eqs. (6.44) and (6.52) in \cite{Santalo} often result in negative $\Pr(l(u,v)\subset D)$.
As the line segment length increases, the relative error increases.
This is because $\Pr(l(u,v)\subset D)$ becomes small and small errors result in large relative errors.

\begin{figure}[htb] 
\begin{center} 
\includegraphics[width=8cm,clip]{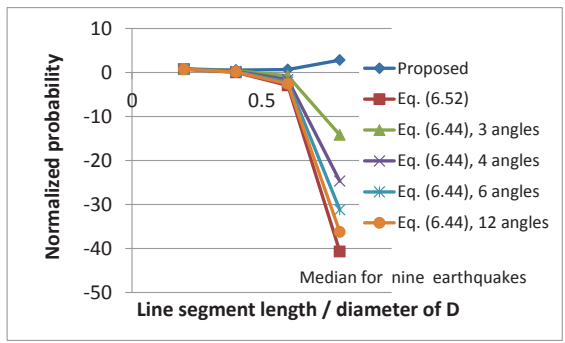} 
\caption{Normalized $\Pr(l(u,v)\subset D)$ of each approximation} 
\label{prob_include} 
\end{center} 
\end{figure}

\end{document}